\documentclass[12pt,a4paper]{article}
\usepackage{amsthm,amsfonts,amsmath,amssymb,bookmark}

\theoremstyle{plain}
\newtheorem{theorem}{Theorem}
\newtheorem{lemma}[theorem]{Lemma}

\theoremstyle{definition}

\theoremstyle{remark}

\def\leqslant{\le}
\def\bq{\begin{eqnarray}}
\def\eq{\end{eqnarray}}
\def\bqq{\begin{eqnarray*}}
\def\eqq{\end{eqnarray*}}
\def\nn{\nonumber}

\def\eps{\varepsilon}

\def \ess {\rm ess}

\def\R{\mathbb{R}}

\def\cE {\mathcal{E}}
\def \d {{\rm d}}

\title{Existence and blow-up of attractive Gross-Pitaevskii minimizers with general bounded potentials}

\author{Thanh Viet Phan \\
\normalsize{Faculty of Mathematics and Statistics, Ton Duc Thang University} \\
\normalsize{Nguyen Huu Tho 19, Ho Chi Minh City, Vietnam}\\
\normalsize{phanthanhviet@tdt.edu.vn} 
}

\date{\normalsize\today}

\begin{document}

\maketitle


\begin{abstract} The paper is concerned with the existence and blow-up behavior of the minimizers for the 2D attractive Gross-Pitaevskii functional when the interaction strength increases to a critical value. Our results hold for all bounded external potential satisfying some general assumptions. 

\bigskip
   
\end{abstract}


\section{Introduction}

The Bose-Einstein condensation was first observed in 1995 in the Nobel Prize winning works of Cornell, Wieman, and Ketterle \cite{CorWie-95,Ketterle-95} and it has been studied intensively in the last decades due to its various interesting quantum effects such as the superfluidity and the quantized vortices, see e.g. \cite{DalGioPitStr-99,Cooper-08}. It is a remarkable fact that when the interaction is attractive, the condensate may collapse, see e.g.  \cite{BraSacTolHul-95,SacStoHul-98,KagMurShl-98}. In the present paper, we will study the existence and the collapse of the condensate in a specific model.

We will consider a 2D Bose-Einstein condensate with an external potential $V:\R^2\to \R$ and an attractive interaction of strength $a>0$. The condensate is determined by solving the variational problem 
\bq \label{eq:GP}
E_a= \inf_{u\in H^1(\R^2), \|u\|_{L^2} = 1} \cE_a(u)
\eq
where $\cE_a(u)$ is the Gross-Pitaevskii energy functional
$$
\cE_a(u)=\int_{\R^2} \Big( |\nabla u(x)|^2 + V(x)|u(x)|^2 -\frac{a}{2} |u(x)|^4\Big) \d x.
$$
The derivation of the Gross-Pitaevskii functional can be seen in \cite{LewNamRou-15} and references therein. Since $\cE(u)\ge \cE(|u|)$ by the diamagnetic inequality, we can assume that $u\ge 0$ for simplicity. 

When $V=0$, by defining $u_\ell(x)=\ell u(\ell x)$ we have the simple scaling property
$$
\cE_0(u_\ell)= \ell^2 \cE_0(u),  \quad \forall \ell>0.
$$
Therefore, $E_a=-\infty$ if $a>a^*$ and $E_a=0$ if $a\le a^*$, where $a^*$ is the optimal constant in the Gagliardo-Nirenberg inequality:
\bq 
\label{eq:GN} 
\left( \int_{\R^2} |\nabla u(x)|^2 \d x \right)\left( \int_{\R^2} |u(x)|^2 \d x \right) \ge \frac{a^*}{2}  \int_{\R^2} |u(x)|^4 \d x, \quad \forall u\in H^1(\R^2). 
\eq
It is well-known (see e.g. \cite{GidNiNir-81,Weinstein-83,MclSer-87}) that
\bq\label{eq:GN1} 
a^*=\int_{\R^2} |Q|^2 = \int_{\R^2}|\nabla Q|^2 = \frac{1}{2} \int_{\R^2}|Q|^4.
\eq
where $Q\in H^1(\R^2)$ is the unique positive radial solution to the nonlinear equation
\bq \label{eq:Q}
-\Delta Q + Q - Q^3 =0.
\eq 
In particular, when $V=0$, $E_a$ has minimizers if and only if $a=a^*$, and all minimizers are of the form $\beta Q_0(\beta x-x_0)$ with $Q_0=Q/\|Q\|_{L^2}$, $\beta>0$ and $x_0\in \R^2$.

When $V \ne 0$, the situation changes crucially. In \cite{GuoSei-14}, Guo and Seiringer showed that for trapping potentials, i.e.  
$$V(x)\ge 0, \quad \lim_{|x|\to \infty}V(x)=\infty,$$
then $E_a$ has a minimizer if and only if $a<a^*$. Moreover, if $V$ has a unique minimizer $x_0 \in \R^2$ and 
\bq \label{eq:V-GuoSei}
\lim_{x\to x_0} \frac{V(x)-V(x_0)}{|x-x_0|^p} = h_0>0, \quad p>0,
\eq
then the minimizer $u_a$ for $E_a$ satisfies  the blow-up behavior
\bq \label{eq:blowup}
\lim_{a\uparrow a^*} \eps_a u_a(x_0+\eps_a x) =Q_0(x) \quad \text{in~} L^2(\R^2),
\eq
where
\bq \label{eq:epsa}
\eps_a= (a^*-a)^{1/(p+2)} \left( \frac{ p h_0}{2} \int_{\R^2} |x|^p |Q(x)|^2 \d x \right)^{-1/(p+2)} .
\eq

In fact, the authors in \cite{GuoSei-14} proved a generalization of \eqref{eq:blowup} when $V$ has finite minimizers, and their result has been extended to other kinds of trapping potentials, see \cite{DenGuoLu-15,GuoWangZengZhou-15,GuoZenZho-16}.

In \cite{Phan-17a}, we proved that if $V$ has a nontrivial negative part, i.e.  
$$
0 \not\equiv \min\{V,0\} \in L^p(\R^2)+L^q(\R^2), \quad 1<p<q<\infty,
$$
then $E_a$ has a minimizer if $a\in (a_*,a^*)$ for some constant $a_*<a^*$. Moreover, if $V$ has a single singular point $x_0$, e.g. 
$$V(x)=-\frac{1}{|x-x_0|^p},\quad 0<p<2,$$ 
then a blow-up result similar to \eqref{eq:blowup}-\eqref{eq:epsa} holds true. 

In the present paper, we are interested in bounded potentials. An important example is the periodic potential, e.g.
$$V(x+z)=V(x), \quad \forall z\in \mathbb{Z}^2,$$
which has been observed in many experiments, see e.g. \cite{BerMol-98,BCCDKP-01,MorObe-06}. The existence and blow-up property of the minimizers for $E_a$ when $a\uparrow a^*$ for continuous, periodic potentials has been solved in \cite{WanZha-16}.

Our aim is to establish the existence and blow-up results for a very general class of bounded potentials, without assuming the periodicity. Our main result is 

\begin{theorem}  \label{thm:main} Let $V\in L^\infty(\R^2,\R)$ satisfy the following two conditions:
\begin{itemize}

\item[\rm (V1)] $\inf \sigma(-\Delta+V)> {\ess}\, {\inf}\, V$;

\item[\rm (V2)] There exists $\eps>0$ such that for all $u\in H^1(\R^2)$ satisfying $\|u\|_{L^2}=1$ and
\bq\label{eq:V2u}\inf_{x\in \R^2} (V*|u|^2)(x)<{\ess}\,{\inf}\,V+\eps,
\eq
the function $x \mapsto (V*|u|^2)(x)$ has (at least) a (global) minimizer on $\R^2$.
\end{itemize}
Then we have the following conclusions.

\begin{itemize}

\item[(i)] (Nonexistence) $E_a=-\infty$ if $a>a^*$ and $E_{a^*}={\ess} \,{\inf} \, V$ but it has no minimizer.

\item [(ii)] (Existence) There exists a constant $a_*\in (0,a^*)$ such that for all $a_*<a<a^*$, the variational problem $E_a$ in \eqref{eq:GP} has (at least) a minimizer. Moreover, if $\{u_n\}$ is a minimizing sequence for $E_a$, then there exist a subsequence of $\{u_n\}$ and a sequence $\{y_n\}\subset \R^2$ such that $u_n(.-y_n)$ converges strongly in $H^1(\R^2)$ to a minimizer for $E_a$.

\item [(iii)] (Blow-up) Assume $a_n \uparrow a^*$ and let $u_n$ be a minimizer for $E_{a_n}$. Then 
$$
\eps_n:= \|\nabla u_n\|_{L^2}^{-1} \to 0.
$$
Moreover, up to a subsequence of $\{u_n\}$, there exists a sequence $\{x_n\} \in \R^2$ such that  
$$
\lim_{n\to \infty}\eps_n u_{n}(\eps_n (x-x_n)) = Q_0(x)\quad \text{strongly in~}H^1(\R^2).
$$
\end{itemize}
\end{theorem}

Let us explain the motivation of the above conditions.

\begin{itemize}

\item (V1) is necessary because if 
$$\inf \sigma(-\Delta+V)={\ess}\, {\inf}\, V$$
then $E_a={\ess}\, {\inf}\, V$ but it has no minimizer for all $a< a^*$.  Here, as usual, we denote by ${\ess}\, {\inf}\, V$ the essentially infimum of $V$ and $\inf \sigma(-\Delta+V)$ the infimum of the spectrum of $-\Delta+V$, i.e.
$$ {\ess}\, {\inf}\, V = \sup \{ s\in \R \,|\, V(x) \ge s \text{~for a.e.~} x\in \R^2\},$$
$$ \inf \sigma(-\Delta+V) = \inf_{\|u\|_{L^2}=1} \int_{\R^2} \Big(|\nabla u|^2 + V|u|^2 \Big).$$

\item (V2) is motivated from the fact that if $E_a$ has a minimizer $u$, then since 
$$ \cE_a(u(y-.))\ge \cE_a(u)$$
we get
$$ (V*|u|^2)(y) =\int V(x) |u(y-x)|^2 \d x \ge \int V(x)|u(x)|^2 \d x= (V*|u|^2)(0), \quad \forall y\in \R^2.$$
Note that the function $x\mapsto (V*|u|^2)(x)$ is uniformly continuous and bounded because $V\in L^\infty(\R^2)$ and $|u|^2\in L^1(\R^2)$. 

If $V$ is periodic, then $V*|u|^2$ is also periodic, and hence $(V2)$ holds true. 

Moreover, $(V2)$ holds true for many other functions, for example the sine cardinal (or sampling) function \cite{Stenger-93}
$$
{\rm sinc}(x)=\frac{\sin(|x|)}{|x|}.
$$
Indeed, $\inf_{x\in \R^2} {\rm sinc}(x)\approx -0.217<0$ and ${\rm sinc}(x)\to 0$ as $|x|\to \infty$. Therefore, if $u\in H^1(\R^2)$ satisfies \eqref{eq:V2u} with $\eps>0$ small enough, then the function $f(x)=({\rm sinc}*|u|^2)(x)$ has a global minimizer on $\R^2$  because $f$ is continuous, $f(x)\to 0$ as $|x|\to \infty$ and, by   \eqref{eq:V2u},
$$
\inf_{x\in \R^2} f(x) \le \inf_{x\in \R^2} {\rm sinc}(x)+\eps<0
$$

\end{itemize}

In Section \ref{sec:nonexistence} and \ref{sec:existence}, we prove the nonexistence and existence part using the concentration-compactness method of Lions \cite{Lions-84,Lions-84b}. In Section \ref{sec:blow-up}, we prove the blow-up property by showing that, up to an appropriate modification, the sequence $\{u_n\}$ forms a minimizing sequence for the Gagliardo-Nirenberg inequality \eqref{eq:GN}. 

\section{Nonexistence} \label{sec:nonexistence}

In this section, we prove the nonexistence part of Theorem \ref{thm:main}. As a preliminary step, we recall the following result

\begin{lemma} \label{lem:ea->0} For all $V\in L^\infty(\R^2,\R)$, then 
\bq 
\label{eq:lima}\lim_{a\uparrow a^*}E_a =E_{a^*} = \ess \inf V.
\eq
\end{lemma}

\begin{proof} The proof of \eqref{eq:lima} is similar to that in \cite{GuoSei-14,Phan-17a} and we recall it below for the reader's convenience. As in \cite{GuoSei-14} we use the trial function
$$
u(x)=A_\ell \varphi(x-x_0) Q_0(\ell(x-x_0))\ell
$$
where $0\le \varphi \in C_c^\infty (\R^2)$, $\varphi(x)=1$ for $|x|\le 1$, and $A_\ell>0$ is a normalizing factor. Since both $Q_0$ and $|\nabla Q_0|$ are exponentially decay (see \cite[Proposition 4.1]{GidNiNir-81}), we have
\begin{align*}
{A_\ell^{-2}} &=  \int_{\R^2} \varphi^2(x-x_0) |Q_0(\ell(x-x_0))|^2 \ell^2 \d x\\
&=\int_{\R^2} \varphi^2(x/\ell) |Q_0(x)|^2 \d x = 1 + O(\ell^{-\infty})
\end{align*}
and 
\begin{align*}
\int_{\R^2} |\nabla u|^2 - \frac{a}{2}\int_{\R^2}|u|^4 &= \ell^2 \left( \int_{\R^2} |\nabla Q_0|^2 - \frac{a}{2}\int_{\R^2}|Q_0|^4  \right) +O(\ell^{-\infty}) \\
& =  \frac{\ell^2(a^*-a)}{2}  \int_{\R^2}|Q_0|^4 +O(\ell^{-\infty}).
\end{align*}
Here $O(\ell^{-\infty})$ means that this quantity converges to $0$ faster than $\ell^{-k}$ when $\ell\to \infty$ for all $k=1,2,...$ Moreover, when $\ell\to \infty$, since $x\mapsto V(x) |\varphi(x-x_0)|^2$ is integrable and $\ell^2 |Q_0(\ell(x-x_0))|^2$ converges weakly to Dirac-delta function at $x_0$ when $\ell\to \infty$, we have
\begin{align*}
\int_{\R^2} V |u|^2  = {|A_\ell|^2} \int_{\R^2} V(x) |\varphi(x-x_0)|^2  |Q_0(\ell(x-x_0))|^2 \ell^2 \d x \to V(x_0) 
\end{align*}
for a.e. $x_0\in \R^2$. Thus in summary,
$$
E_a \le \cE_a(u) \le \frac{\ell^2(a^*-a)}{2}  \int_{\R^2}|Q_0|^4 + V(x_0)+ O(\ell^{-\infty})
$$
for a.e. $x_0\in \R^2$. By choosing $\ell=(a^*-a)^{-1/4}$ and optimizing over $x_0$, we obtain that
\bqq
\limsup_{a\uparrow a^*}E_a \le {\ess}\, {\inf} V. 
\eqq
On the other hand, by the Gagliardo-Nirenberg inequality \eqref{eq:GN}, $E(a)\ge E(a^*)\ge \ess\inf V$. Thus \eqref{eq:lima} holds true, i.e.
$$ \liminf_{a\uparrow a^*}E_a=E_{a^*}= {\ess}\, {\inf} \,V.$$
\end{proof}

From Lemma \ref{lem:ea->0}, it is easy to deduce the nonexistence part of Theorem \ref{thm:main}. 

\begin{proof}[Proof of Theorem \ref{thm:main} (Nonexistence part)] By assumption (V1), we have $V\not\equiv {\rm constant}$. 

If $E_{a^*}$ has a minimizer $u^*$, then by  \eqref{eq:lima}, we have
$$
{\ess}\,{\inf} \,V = E_{a^*}=\cE_{a^*}(u^*) =  \int_{\R^2} V|u^*|^2+ \Big[ \int_{\R^2} |\nabla u^*|^2-\frac{a^*}{2}\int_{\R^2} |u^*|^4 \Big].
$$
Using
$$
{\ess}\,{\inf} \,V\le \int_{\R^2} V|u^*|^2
$$
and the Gagliardo-Nirenberg inequality \eqref{eq:GN}, we deduce that
\bq \label{eq:abc}
{\ess}\,{\inf} \,V= \int_{\R^2} V|u^*|^2
\eq
and 
\bq \label{xyz}
\int_{\R^2} |\nabla u^*|^2-\frac{a^*}{2}\int_{\R^2} |u^*|^4 =0
\eq
From \eqref{xyz}, we see that $u^*$ is an optimizer for the interpolation inequality \eqref{eq:GN}. This implies that $u$ is equal to $Q_0$ up to translations and dilations. Since $Q_0(x)>0$, we have $|u_0(x)|^2>0$ for all $x\in \R^2$. But in this case \eqref{eq:abc} can not occur except when $V$ is a constant function. This contradiction implies that $E_{a^*}$ has no mimimizer.

Next, we show that $E_a=-\infty$ if $a>a^*$. From \eqref{eq:GN1} and the definition $Q_0=Q/\|Q\|$ we have
$$
\int_{\R^2} |\nabla Q_0|^2  = \frac{a^*}{2}\int_{\R^2}|Q_0|^4=1.
$$
Therefore, with the choice $u_\ell(x)=\ell Q_0(\ell x)$ we get 
\begin{align*}
E_a \le \cE_{a}(u_\ell)&=\ell^2 \int_{\R^2} |\nabla Q_0|^2 + \int_{\R^2} V(./\ell)|Q_0|^2 - \frac{a\ell^2}{2}\int_{\R^2}|Q_0|^4 \\
&\le \ell^2 \Big( 1-\frac{a}{a^*}\Big) + {\ess} \,{ \sup}\, V.
\end{align*}
Since $V$ is bounded and $a>a^*$, we can take $\ell\to \infty$ to conclude that $E_a=-\infty$.
\end{proof}

\section{Existence} \label{sec:existence}

Now we turn to the  existence result in Theorem \ref{thm:main}. The key tool is concentration-compactness argument. For the reader's convenience, we recall the following standard result, which essentially goes back to Lions \cite{Lions-84,Lions-84b}.

\begin{lemma}[Concentration-compactness]\label{lem:concom} Let $N\ge 1$. Let $\{u_n\}$ be a bounded sequence in $H^1(\R^N)$ with $\|u_n\|_{L^2}=1$. Then there exists a subsequence (still denoted by $\{u_n\}$ for simplicity) such that  one of the following cases occurs:

\begin{itemize}

\item[(i)] {\rm (Compactness)} There exists a sequence $\{x_n\}\subset \R^N$ such that $u_n(.+x_n)$ converges strongly in $L^p(\R^N)$ for all $p\in [2,2^*)$. 
\item[(ii)] {\rm (Vanishing)} $u_n\to 0$ strongly in $L^p$ for all $p\in (2,2^*)$. 
\item[(iii)] {\rm (Dichotomy)} There exist $\lambda\in (0,1)$ and two sequences $\{u_n^{(1)}\}$, $\{u_n^{(2)}\}$ in $H^1(\R^N)$ such that 
\[ 
\left \{
\begin{aligned}
  & \lim_{n\to \infty} \int_{\R^N} |u_n^{(1)}|^2 = \lambda,  \quad \lim_{n\to \infty} \int_{\R^N} |u_n^{(2)}|^2 = 1-\lambda,  \\
  & \lim_{n\to \infty} {\rm dist} ({\rm supp} (u_n^{(1)}), {\rm supp} (u_n^{(2)})) =+\infty;\\
  & \lim_{n\to \infty} \|u_n - u_n^{(1)}-u_n^{(2)}\|_{L^p} = 0  , \quad  \forall p\in [2,2^*); \\
  & \liminf_{n\to \infty} \int_{\R^N} (|\nabla u_n|^2 - |\nabla u_n^{(1)}|^2 - |\nabla u_n^{(2)}|^2) \ge 0.
  \end{aligned}
 \right.
\]
\end{itemize}

\end{lemma}
Here $2^*$ is the critical power in Sobolev's embedding, i.e. $2^*=2N/(N-2)$ if $N\ge 3$ and $2^*=+\infty$ if $N\le 2$.  
\begin{proof} The result is essentially taken from  \cite[Lemma III.1]{Lions-84}, with some minor modifications that we explain below.

(i) The original notion of the compactness case in \cite[Lemma I.1]{Lions-84} reads
\bq \label{eq:compactness}
\lim_{R\to \infty} \int_{|x|\le R} |u_n(x+x_n)|^2 \d x =1.
\eq
Since $u_n(.+x_n)$ is bounded in $H^1(\R^N)$, up to subsequences $u_n(.+x_n)$ converges weakly to some $u$ in $H^1(\R^N)$. This implies that 
\bq \label{eq:compactness1}\int\limits_{|x| \leqslant R} {{{\left| {{u_n}(. + {x_n})} \right|}^2}}  = \int\limits_{|x| \leqslant R} {{{\left| u \right|}^2}}  + \int\limits_{|x| \leqslant R} {{{\left| {{u_n}(. + {x_n}) - u} \right|}^2}}  + o{(1)_{n \to  + \infty }}.\eq
Moreover, 
\bq \label{eq:compactness2}
{\chi _{\{ |x| \leqslant R\} }}{u_n}(. + {x_n}) \to {\chi _{\{ |x| \leqslant R\} }}u
\eq
strongly in $L^2$, as explained in \cite[Section 8.6]{Lieb-01}.

From (\ref{eq:compactness}), (\ref{eq:compactness1}),(\ref{eq:compactness2}) we obtain that $\|u\|_{L^2}=1$. Hence $u_n(.+x_n)$ converges strongly in $L^2(\R^N)$. 

For any $2\le p \le q$, using the interpolation inequality and the Sobolev's embedding, we have 
\bqq
\|u_n(.+x_n)-u\|_{L^p} &\le& \|u_n(.+x_n)-u\|_{L^2}^\alpha\|u_n(.+x_n)-u\|_{L^q}^{1-\alpha}\\
&\le& \|u_n(.+x_n)-u\|_{L^2}^\alpha \|u_n(.+x_n)-u\|_{H^1}^{1-\alpha}\\
&\le& C\|u_n(.+x_n)-u\|_{L^2}^\alpha,
\eqq
where $\frac{1}{p}=\frac{\alpha}{2}+\frac{1-\alpha}{q}$, $0\le \alpha\le 1.$

Therefore, $u_n(.+x_n)$ converges strongly in $L^p$. 

(ii) The original notion of the vanishing case in \cite[Lemma I.1]{Lions-84} reads
$$
\lim_{R\to \infty} \sup_{y\in \R^N} \int_{|x|\le R} |u_n(x+y)|^2 \d x =0.
$$
This and the boundedness in $H^1$ implies that $u_n\to 0$ strongly in $L^p(\R^N)$ for all $p\in (2,2^*)$, as explained in \cite[Lemma I.1]{Lions-84b}.

(iii) In the dichotomy case, 
the original statement in \cite[Lemma III.1]{Lions-84} has a parameter $\eps\to 0$, but this parameter can be relaxed by a standard Cantor's diagonal argument. 
\end{proof}

\begin{proof}[Proof of Theorem \ref{thm:main} (Existence part)] From Lemma \ref{lem:ea->0} and Assumptions (V1)-(V2), we can find $a_*\in (0,a^*)$ such that 
\bq \label{eq:Ea<sigma}
E_a<\min\Big(\inf \sigma(-\Delta+V), {\ess}\,{\inf}\,  V+\eps\Big), \quad \forall a\in (a_*,a^*).
\eq
where $\eps>0$ is the constant in \eqref{eq:V2u}. We will prove that $E_a$ has a minimizer for all  $a\in (a_*,a^*)$.

Using the boundedness of $V$ and the Gagliardo-Nirenberg inequality \eqref{eq:GN}, we get 
$$
\cE_a(u) \ge \Big(1 - \frac{a}{a^*}\Big) \int_{\R^2}|\nabla u|^2 - \|V\|_{L^\infty}, \quad \forall u\in H^1(\R^2), \int |u|^2=1.
$$
Thus $E_a>-\infty$ and if $\{u_n\}$ is a minimizing sequence for $E_a$, then it is bounded uniformly in $H^1(\R^2)$. By Concentration-Compactness Lemma \ref{lem:concom}, up to a subsequence of $\{u_n\}$, we will obtain either compactness, vanishing, or dichotomy. We will show that the vanishing and dichotomy can not happen.
\\

{\bf No vanishing.} If $\{u_n\}$ is vanishing, then $\|u_n\|_{L^p}\to 0$ for all $p\in (2,\infty)$. Therefore,  
$$
E_a = \lim_{n\to \infty} \cE_a(u_n) \ge \liminf_{n\to \infty} \int (|\nabla u_n|^2 + V|u_n|^2) \ge \inf \sigma(-\Delta+V).
$$
However, this contradicts to the inequality $E_a < \inf \sigma(-\Delta+V)$ in \eqref{eq:Ea<sigma}.\\

{\bf No dichotomy.} Assume the dichotomy occurs. Let $\{u_n^{(1)}\}$, $\{u_n^{(2)}\}$ be the two corresponding sequences. Let us show that
\bq \label{eq:decomp-Ea-un}
\liminf_{n\to \infty} (\cE_a(u_n)- \cE_a(u_n^{(1)}) - \cE_a(u_n^{(2)})) \ge 0.
\eq
Indeed, by Lemma \ref{lem:concom} (iii) we already have
\bq\label{eq:12u1u2}
\liminf_{n\to \infty} \int_{\R^2} (|\nabla u_n|^2 - |\nabla u_n^{(1)}|^2 - |\nabla u_n^{(2)}|^2) \ge 0.
\eq
From (\ref{eq:12u1u2}) and since ${u_n }$ is a bounded sequence in $H^1(\mathbb{R}^N)$, we obtain that $u_n^{(1)}$ and $u_n^{(2)}$ are also bounded in $H^1(\mathbb{R}^N)$. Thus by using the Sobolev's embedding, there exists a constant $C$ such that 
\bqq
max\{\|u_n\|_{L^p},\|u_n^{(1)}\|_{L^p},\|u_n^{(2)}\|_{L^p}\}\le C ~~ ,\forall p\ge 2.
\eqq
Moreover, since $u_n^{(1)}$ and $u_n^{(2)}$ have disjoint supports and $\|u_n-u_n^{(1)}-u_n^{(2)}\|_{L^p}\to 0$ when $n\to \infty$ for all $p\in [2,\infty)$, we find that
\begin{align*}
\int_{\R^2} (|u_n|^4-|u_n^{(1)}|^4-|u_n^{(2)}|^4) =  \int_{\R^2}  (|u_n|^4-|u_n^{(1)}+u_n^{(2)}|^4) \to 0.
\end{align*}
Similarly, since $V$ is bounded, we get 
$$
\int_{\R^2} V (|u_n|^2-|u_n^{(1)}|^2-|u_n^{(2)}|^2) = \int_{\R^2} V (|u_n|^2-|u_n^{(1)}+u_n^{(2)}|^2)  \to  0.
$$
Thus \eqref{eq:decomp-Ea-un} holds true. 

Next, using $\|u_n^{(1)}\|_{L^2}^2\to \lambda$, we obtain
\begin{align*}
\cE_a(u_n^{(1)}) &= \int_{\R^2} \Big( |\nabla u_n^{(1)}|^2 + V|u_n^{(1)}|^2  \Big) -\frac{a}{2} \int_{\R^2} |u_n^{(1)}|^4 + o(1)_{n\to \infty}\\
&= \Big(1-\|u_n^{(1)}\|_{L^2}^2\Big)  \int_{\R^2} \Big( |\nabla u_n^{(1)}|^2 + V|u_n^{(1)}|^2  \Big)  \\
&\qquad \qquad + \|u_n^{(1)}\|_{L^2}^4 \cE_a\left( \frac{u_n^{(1)}}{\|u_n^{(1)}\|_{L^2}}\right) + o(1)_{n\to \infty} \\
&\ge (1-\lambda)\lambda \inf \sigma(-\Delta+V) + \lambda^2 E_a + o(1)_{n\to \infty}. 
\end{align*}
Similarly, using $\|u_n^{(2)}\|_{L^2}^2\to 1-\lambda$, we get
\begin{align*}
\cE_a( u_n^{(2)}) \ge (1-\lambda)\lambda \inf \sigma(-\Delta+V) + (1-\lambda)^2 E_a + o(1)_{n\to \infty}. 
\end{align*}
Inserting these estimates into \eqref{eq:decomp-Ea-un}, we find that
\begin{align*} \cE_a(u_n) &= \cE_a(u_n^{(1)}) + \cE_a(u_n^{(2)}) + o(1)_{n\to \infty} \\
&\ge 2(1-\lambda)\lambda \inf \sigma(-\Delta+V) + \Big(\lambda^2 + (1-\lambda)^2\Big) E_a + o(1)_{n\to \infty}.
\end{align*}
Taking $n\to \infty$ we obtain
\begin{align*} E_a \ge 2(1-\lambda)\lambda \inf \sigma(-\Delta+V) + \Big(\lambda^2 + (1-\lambda)^2\Big) E_a.
\end{align*}
Since $1>\lambda>0$, this leads to $E_a\ge  \inf \sigma(-\Delta+V)$, which contradict to  \eqref{eq:Ea<sigma}.\\

{\bf Compactness.} Thus from Lemma \ref{lem:concom}  we conclude that the compactness occurs, i.e. there exists a sequence $\{x_n\}\subset \R^2$ such that $u_n(.+x_n)$ converges to some $u_0$ weakly in $H^1(\R^2)$ and strongly in $L^p(\R^2)$ for all $p\in [2,\infty)$. Then we have 
$$u_0\in H^1(\R^2), \int |u_0|^2=1$$
and
\begin{align}
\int_{\R^2} |\nabla u_n|^2  &= \int_{\R^2} |\nabla u_n(x+x_n)|^2 \d x \ge  \int_{\R^2} |\nabla u_0|^2 + o(1)_{n\to \infty} ,\label{eq:Fatou-gradient}\\
\int_{\R^2} |u_n|^4  &= \int_{\R^2} |u_n(x+x_n)|^4 \d x =  \int_{\R^2} |u_0|^4 + o(1)_{n\to \infty}.
\end{align}
Moreover, since $V$ is bounded and $u_n(.+x_n)\to u$ strongly in $L^2(\R^2)$, we can write
$$ \int_{\R^2} V |u_n|^2 = \int_{\R^2} V(x+x_n) |u_n(x+x_n)|^2 \d x = \int_{\R^2} V(x+x_n) |u_0(x)|^2 \d x + o(1)_{n\to \infty}.
$$
In summary, 
$$
\cE_a(u_n) \ge \int_{\R^2} \Big(|\nabla u_0(x)|^2 + V(x+ x_n)|u_0(x)|^2 - \frac{a}{2} |u_0(x)|^4 \Big) \d x   + o(1)_{n\to \infty}.
$$
Since $u_n$ is a minimizing sequence, we conclude that
\begin{align} \label{eq:cEa>=cEau0}
E_a \ge  \int_{\R^2} \Big(|\nabla u_0(x)|^2 + V(x+ x_n)|u_0(x)|^2 - \frac{a}{2} |u_0(x)|^4 \Big) \d x   + o(1)_{n\to \infty}.
\end{align}

{\bf Conclusion.} From \eqref{eq:cEa>=cEau0} and the Gagliardo-Nirenberg inequality \eqref{eq:GN}, we obtain
$$
\inf_{y\in \R^2}  \int_{\R^2}V(x+ y)|u_0(x)|^2 \d x \le \liminf_{n\to \infty} \int_{\R^2}  V(x+ x_n)|u_0(x)|^2 \d x \le E_a.
$$
Combining with \eqref{eq:Ea<sigma} we find that
\bq \label{eq:V2ua}
\inf_{y\in \R^2}  \int_{\R^2}V(x+ y)|u_0(x)|^2 \d x \le {\ess}\,{\inf} \, V+\eps
\eq
where $\eps>0$ is the constant in \eqref{eq:V2u}. 

To use Assumption (V2), we introduce the function 
$$v(x):=u_0(-x)$$
which satisfies
\begin{align} \label{eq:V2ub}
(V*|v|^2)(y)&=\int_{\R^2} V(y-x)|v(x)|^2 \d x = \int_{\R^2} V(y+x)|v(-x)|^2 \d x \nn \\
&=  \int_{\R^2} V(y+x)|u_0(x)|^2 \d x.
\end{align}
Thus \eqref{eq:V2ua} is equivalent to
$$
\inf_{y\in \R^2} (V*|v|^2)(y) \le {\ess}\,{\inf} \, V+\eps.
$$
Of course  $v\in H^1(\R^2)$, $\|v\|_{L^2}=1$. Therefore, by Assumption (V2), the function $y\mapsto (V*|v|^2)(y)$ has a global minimizer $x_0\in \R^2$. By \eqref{eq:V2ub}, we obtain
\bq \label{eq:x0-glmin}
\int_{\R^2}V(x+y) |u(x)|^2 \d x \ge \int_{\R^2}V(x+x_0) |u(x)|^2 \d x, \quad \forall y\in \R^2.
\eq

Finally, combining \eqref{eq:x0-glmin} and \eqref{eq:cEa>=cEau0}, we find 
\begin{align*}
E_a &\ge  \int_{\R^2} \Big(|\nabla u_0(x)|^2 + V(x+ x_0)|u_0(x)|^2 - \frac{a}{2} |u_0(x)|^4 \Big) \d x   + o(1)_{n\to \infty}\\
&= \cE_a(u_0(.-x_0))+o(1)_{n\to \infty}.
\end{align*}
Thus, by passing $n\to \infty$, we conclude that $u_0(.-x_0)$ is a minimizer for $E_a$. 

We have already had that $u_n(.+x_n)\to u_0$ weakly in $H^1(\R^2)$ and strongly in $L^2(\R^2)$. Moreover, from the above proof, we see that the equality must occurs in \eqref{eq:Fatou-gradient}, i.e. 
$$
\lim_{n\to \infty} \int_{\R^2} |\nabla u_n(x_n+x)|^2 \d x =  \int_{\R^2} |\nabla u_0(x)|^2 \d x.
$$
Thus we conclude that $u_n(.+x_n)\to u_0$ strongly in $H^1(\R^2)$. Equivalently, $u_n(.+x_n-x_0)\to u(.-x_0)$ strongly in $H^1(\R^2)$.
\end{proof}

\section{Blow-up} \label{sec:blow-up}

In this section, we prove the blow-up part of Theorem \ref{thm:main}. In the original paper of Guo and Seiringer \cite{GuoSei-14}, the blow-up result was proved by a careful analysis of the Euler-Lagrange equation associated to the variational problem $E_a$. This approach has been followed by many other authors, e.g. \cite{DenGuoLu-15,GuoWangZengZhou-15,GuoZenZho-16,WanZha-16}. Here we represent another, much simpler approach which does not use the Euler-Lagrange equation at all. 


The key tool of our approach is the compactness of minimizing sequences for the  Gagliardo-Nirenberg inequality \eqref{eq:GN}.

\begin{lemma} \label{lem:GN-bl} Let $\{f_n\}$ be  a bounded sequence in $H^1(\R^2)$ satisfying
$$ \|\nabla f_n\|_{L^2}= \|f_n\|_{L^2}=1, \quad \frac{a^*}{2} \|f_n\|_{L^4}^4 \to 1.$$
Then there exist a subsequence of $\{f_n\}$ and a sequence $\{x_n\}\subset \R^2$ such that 
$$ \lim_{n\to \infty}f_n(.-x_n)=Q_0\quad \text{strongly in~}H^1(\R^2).$$
\end{lemma}

\begin{proof} Let us apply Concentration-Compactness Lemma \ref{lem:concom} to the sequence $\{f_n\}$.  The vanishing does not occurs because $\|f_n\|_{L^4}^4 \to 2/a^*>0$. Now we assume the dichotomy occurs and let $\{f_n^{(1)}\}$, $\{f_n^{(2)}\}$ be two corresponding sequences. From Lemma \ref{lem:concom}, we have 
\begin{align*} 
\|f_n^{(1)}\|_{L^2}^2&=\lambda, \quad \|f_n^{(2)}\|_{L^2}^2=1-\lambda, \quad \lambda\in (0,1),\\
1=\|\nabla f_n\|_{L^2}^2 &\ge \|\nabla f_n^{(1)}\|_{L^2}^2+ \|\nabla f_n^{(2)}\|_{L^2}^2 + o(1)_{n\to \infty} \\
1=\frac{a^*}{2}\|f_n\|_{L^4}^4+o(1)_{n\to \infty} &= \frac{a^*}{2} \Big( \|f_n^{(1)}\|_{L^4}^4+\|f_n^{(2)}\|_{L^4}^4 \Big)+o(1)_{n\to \infty}.
\end{align*}
On the other hand, by the Gagliardo-Nirenberg inequality \eqref{eq:GN},
$$
\|\nabla f_n^{(1)} \|_{L^2}^2 \ge  \frac{a^*}{2\lambda} \|f_n^{(1)}\|_{L^4}^4 , \quad  \|\nabla f_n^{(2)} \|_{L^2}^2 \ge \frac{a^*}{2(1-\lambda)} \|f_n^{(2)}\|_{L^4}^4.
$$
Combining these estimates, we find that
\begin{align*}
1=\|\nabla f_n\|_{L^2}^2 &\ge \liminf_{n\to \infty} \Big(  \|\nabla f_n^{(1)}\|_{L^2}^2+ \|\nabla f_n^{(2)}\|_{L^2}^2 \Big)\\
&\ge  \min\left\{ \frac{1}{\lambda}, \frac{1}{1-\lambda}\right\}  \frac{a^*}{2} \liminf_{n\to \infty} \Big( \|f_n^{(1)}\|_{L^4}^4+\|f_n^{(2)}\|_{L^4}^4 \Big) \\
&=\min\left\{ \frac{1}{\lambda}, \frac{1}{1-\lambda}\right\} .
\end{align*}
However, this is a contradiction because $0<\lambda<1$. Thus the dichotomy does not occur.

Therefore, we obtain the compactness in Lemma \ref{lem:concom}, i.e. there exist a subsequence of $\{f_n\}$ and a sequence $\{x_n\}\subset \R^2$ such that $f_n(.-x_n)$ converges to some function $f$ weakly in $H^1(\R^2)$ and strongly in $L^p(\R^2)$ for all $p\in [2,\infty)$. Then we have
\bq \label{eq:bl-NG}
\|\nabla f\|_{L^2}^2 \le \lim_{n\to \infty} \|\nabla f_n\|_{L^2}^2 =1 = \lim_{n\to \infty}  \frac{a^*}{2} \|f_n\|_{L^4}^4 = \frac{a^*}{2} \|f\|_{L^4}^4.
\eq
In view of the Gagliardo-Nirenberg inequality \eqref{eq:GN} and the constraint $\|f\|_{L^2}=1$, we conclude that
$$
\|\nabla f\|_{L^2}^2= \lim_{n\to \infty} \|\nabla f_n\|_{L^2}^2 =1 
$$
and hence $f_n\to f$ strongly in $H^1(\R^2)$. Moreover, since $f$ is a minimizer for \eqref{eq:bl-NG} and $Q_0$ is the unique minimizer for \eqref{eq:bl-NG} up to translations and dilations, we obtain
$$f(x)=\beta Q_0(\beta x-x_0)$$
for some constant  $\beta>0$ and $x_0\in \R^2$. From (\ref{eq:GN1}) and since $\|\nabla f\|_{L^2}^2=1$, we get $\beta=1$.

Thus $f_n(.-x_n+x_0)\to f(.+x_0)=Q_0$ strongly in $H^1(\R^2)$. 
\end{proof}

Now we finish the proof of Theorem \ref{thm:main}.
\begin{proof}[Proof of Theorem \ref{thm:main} (Blow-up part)] Let $a_n \uparrow a$ and let $u_n$ be a minimizer for $E_{a_n}$. Let us show that 
\bq \label{eq:bl1}
\eps_n  := \|\nabla u_n\|_{L^2}^{-1} \to 0
\eq
as $n\to \infty$. We assume by contradiction that $u_n$ has a subsequence which is bounded in $H^1(\R^2)$. Then by applying Concentration-Compactness Lemma \ref{lem:concom} to this subsequence and following the proof of the existence part, we can show that up to subsequences and translations, $u_n$ converges strongly to a minimizer of $E_{a^*}$. However, it is in contradiction to the fact that $E_{a^*}$ has no minimizer in Lemma \ref{lem:ea->0}.

On the other hand, note that
$$
\cE_{a_n}(u_n) \ge \int_{\R^2} \Big(|\nabla u_n|^2 - \frac{a_n}{2}|u_n|^4 \Big) + {\ess}\, {\inf}\, V \ge {\ess}\, {\inf}\, V 
$$
and, by Lemma \ref{lem:ea->0},
$$
\cE_{a_n}(u_n) = E_{a_n}\to  {\ess}\, {\inf}\, V.
$$
as $n\to \infty$. Thus 
\bq \label{eq:bl2}
\lim_{n\to \infty} \int_{\R^2} \Big(|\nabla u_n|^2 - \frac{a_n}{2}|u_n|^4 \Big)=0.
\eq

From \eqref{eq:bl1} and \eqref{eq:bl2}, we can rescale and find that the sequence
$$
f_n(x):= \eps_n u_n(\eps_n x)
$$
satisfies
$$
\|f_n\|_{L^2}=\|\nabla f_n\|_{L^2}=1, \quad \frac{a^*}{2} \|f_n\|_{L^4}^4 \to 1.
$$
Thus we can apply Lemma \ref{lem:GN-bl} to the sequence $\{f_n\}$. The conclusion is that a subsequence of $\{f_n\}$, there exists a sequence $\{x_n\}\subset \R^2$ such that 
$$
\eps_n u_{n}(\eps_n (x-x_n) )  = f_n(x-x_n) \to Q_0(x) \quad \text{strongly in~} H^1(\R^2).
$$
The proof is complete.
\end{proof}

\end{document}